\documentclass[]{amsart}
\usepackage{graphicx} % Required for inserting images
\usepackage{amsthm, amsfonts, amsmath, amssymb}
\usepackage[table]{xcolor}

\usepackage{mathtools}
\DeclarePairedDelimiter\bra{\langle}{\rvert}
\DeclarePairedDelimiter\ket{\lvert}{\rangle}
\DeclareMathOperator{\wt}{wt}
\DeclareMathOperator{\Tr}{Tr}
\DeclareMathOperator{\supp}{supp}
\DeclareMathOperator{\Span}{span}
\DeclareMathOperator{\rank}{rank}
\DeclareMathOperator{\diag}{diag}

\newcommand{\norm}[1]{\left\lVert#1\right\rVert}

\newtheorem{theorem}{Theorem}

\newtheorem{corollary}[theorem]{Corollary}
\newtheorem{example}[theorem]{Example}
\newtheorem{remark}[theorem]{Remark}

\title[EA Codes Outside the Stabilizer Framework]{Entanglement-Assisted Codes Outside\\the Stabilizer Framework}
\author[J.~Defranco and A.~Nemec]{Jaszmine~Defranco$^*$ and Andrew~Nemec$^\dagger$}
\address{Department of Computer Science, University of Texas at Dallas, 890 Franklyn Jenifer Drive, Richardson, TX 75080 USA}
\email{$^*$jaszmine.defranco@utdallas.edu}
\email{$^\dagger$andrew.nemec@utdallas.edu}
\begin{document}

\begin{abstract}
    We show how entanglement-assisted codes can be constructed from arbitrary quantum codes by associating them with quantum codes for erasure channels.
    If a subset of physical qubits is correctable for an erasure error, then it naturally forms the receiver's share of a bipartite state that can be used for entanglement-assisted communications, both in the noiseless and noisy ebit error models.
    In the case of degenerate codes, we show that the receiver's share of the bipartite state can sometimes be compressed, at the cost of potentially reduced error-correction ability in the noisy ebit error model.
    We also give examples of permutation-invariant and XP-stabilizer entanglement-assisted codes, the first outside of the stabilizer and codeword-stabilized frameworks.
\end{abstract}

\maketitle

\section{Introduction}

Quantum error-correcting codes~\cite{Sho1995,Ste1996,CS1996,Got1997,CRSS1998} are necessary for the faithful transmission of quantum information across noisy channels, allowing for reliable quantum communication and quantum computation.
Ever since the first quantum code was demonstrated by Shor~\cite{Sho1995}, quantum coding theory has progressed at a rapid pace, especially after the introduction of the stabilizer formalism~\cite{Got1997,CRSS1998}.
This framework constructs quantum codes using self-orthogonal classical codes, allowing quantum coding theorists to make use of the extensive existing research on classical coding theory.

One important channel in quantum communications is the quantum erasure channel introduced by Grassl et al.~\cite{GBP1997}, in which an arbitrary error affects a known subset of qubits.
A quantum code with minimum distance $d$ is able to correct $d-1$ erasures, compared to $\left\lfloor\left(d-1\right)/2\right\rfloor$ arbitrary errors on unknown subsets of qubits.
The possibility of exploiting this stronger error-correction ability has led to an explosion of recent work on erasure conversion~\cite{WKPT22,KHVLBR2023,KCB23}, in which the physical qubits are chosen in such a way that physical noise causes a detectable leakage outside the qubit subspace; erasure conversion has been implemented using various platforms~\cite{SSTFCE2023,MLPZJCBPPT2023,LHHA+2024,GVRK2025}.
Additionally, codes for quantum erasures arise naturally in quantum cryptography, particularly in quantum secret sharing schemes~\cite{CGL1999,Got2000}, as well as in black hole theory, where connections between the AdS/CFT correspondence, quantum error correction, and quantum secret sharing have been studied~\cite{ADH2015,Har2017,HP2019,AR2019,AP2022}.

In a quantum world, a sender transmitting qubits across a noisy channel to a receiver is not the most general type of protocol.
Some early quantum communications protocols such as superdense coding~\cite{BW1992} and quantum teleportation~\cite{BBCJPW1993} are ``entanglement-assisted'' (EA), in that the sender and the receiver share some number of noiseless entangled qubits (ebits) before the protocol begins.
Bowen first introduced EA quantum error correction (EAQEC), showing how it is possible to send one logical qubit using three physical qubits and correct an arbitrary error to one qubit, assuming that the sender and receiver had pre-shared two ebits~\cite{Bow2002}.

Brun et al.~\cite{BDH2006,BDH2014} fully integrated Bowen's code into a new EA stabilizer framework, which removes the restriction that the classical codes used be self-orthogonal.
This greatly increases the types of classical codes that can be used to construct quantum codes, as long as the participants have some amount of pre-shared entanglement they are able to use as a resource.
Since then, a large body of work surrounding EAQEC has emerged, including EA quantum convolutional codes~\cite{WB2009,WB2010Apr,WB2010Jun}, EA classical codes~\cite{SKMB2010,PG2025}, bounds for EA codes~\cite{LA2018,Gra2021,LTY2024,DDSB2025}, as well as generalizations to the codeword-stabilized framework~\cite{SHB2011,AJH2016} and integration with subsystem and hybrid quantum-classical codes~\cite{HDB2007,KHB2008,SHB2013,NADKV2024}.
Two different error models are typically used: the noiseless ebit error model~\cite{BDH2006}, in which the receiver's share of the bipartite state does not experience noise, and the noisy ebit error model~\cite{LB2012}, in which the receiver's share potentially experiences noise.
See~\cite{BH2013} for an introduction to EAQEC.

The construction of EA codes from existing quantum codes has been studied extensively~\cite{LB2012,GHMR2019,GHMR2021,GHW2022,UM2022}; however, these results are primarily restricted to nondegenerate stabilizer codes and are typically done by a puncturing-like process on subsets of size less than $d$.
In this paper, we will show that any quantum code can easily be turned into an EA code by investigating those subsets that are correctable under the erasure channel.
In particular, we show that the Structure Theorem for replacer codes given recently by Chitambar et al.~\cite{CHKNN2025} naturally identifies both the shared bipartite state for the sender and receiver, as well as the sender's encoding unitary, allowing us to construct the first EA codes outside of the stabilizer and codeword-stabilized frameworks.
We also show that a degenerate correctable set may lead to an EA code with a compressed share for the receiver, at the cost of leaving the receiver's qubits unprotected against errors in the noisy ebit error model.

This paper is organized as follows. In Section 2, we briefly review the basics of both quantum codes for the erasure channel and EAQEC necessary for the rest of the paper. In Section 3, we show how the Structure Theorem of~\cite{CHKNN2025} leads to constructions of EA quantum codes, and we give several new examples of EA codes in the permutation-invariant and XP-stabilizer frameworks. In Section 4, we show how degeneracy in the original code may lead to a smaller entangled share for the receiver, at the cost of potentially losing error-correction protection on the receiver's qubits in the noisy ebit error model. We finish in Section 5 with some concluding remarks and potential future research directions.

\section{Background}

Below we review results from the field of quantum error correction that are necessary for this paper.
Throughout this paper we restrict ourselves to qubit codes.

\subsection{Quantum Error Correction}

A quantum error-correcting code $\mathcal{C}$ is a $K$-dimensional subspace of an $n$-qubit Hilbert space $\mathcal{H}=(\mathbb{C}^2)^{\otimes n}$.
We will write $\{\ket{\widetilde{i}}\}_{i=0}^{K-1}$ as an orthonormal basis of the code $\mathcal{C}$.
We say that a quantum code corrects an error $E\in\mathcal{L\!\left(H\right)}$, where $\mathcal{L\!\left(H\right)}$ is the set of linear operators on $\mathcal{H}$, if there exists a recovery operator $\mathcal{R}$ such that 
\begin{equation}
    \mathcal{R}\!\left(E\ket{\widetilde{\psi}}\!\bra{\widetilde{\psi}}E^\dagger\right)\propto \ket{\widetilde{\psi}}\!\bra{\widetilde{\psi}}
\end{equation}
for every $\ket{\widetilde{\psi}}\in\mathcal{C}$.
An equivalent condition in terms of the orthogonal projector $P$ onto the coding subspace was given by Knill and Laflamme in~\cite{KL1997}:
\begin{theorem}[Error Correction Condition for Quantum Subspace Codes]\label{thm:klcond}
    Let $\mathcal{C}\subseteq\mathcal{H}$ be a quantum code with projector $P$. Then $\mathcal{C}$ corrects the set of errors $\mathcal{E}\subseteq\mathcal{L}\!\left(\mathcal{H}\right)$ if and only if 
    \begin{equation}
        P\left(E_i^\dagger E_j\right)P=\lambda_{ij}P,
    \end{equation}
    for all $E_i,E_j\in\mathcal{E}$, where $\lambda_{ij}\in\mathbb{C}$.
\end{theorem}

The Hamming weight $\wt\!\left(\cdot\right)$ of an operator is the number of physical qubits on which it acts nontrivially.
The most common measure of error-correction ability is given by the minimum distance $d$ of the code\footnote{The minimum distance is typically easier to define using the Knill-Laflamme error detection condition: $d=\min\!\left\{\wt\!\left(E\right)\mid PEP\neq\lambda_E P,\;E\in\mathcal{L}\!\left(\mathcal{H}\right)\right\}$.}: a code with minimum distance $d$ can correct any error $E$ that acts nontrivially on the code with $\wt\!\left(E\right)\leq \left\lfloor\left(d-1\right)/2\right\rfloor$, can detect any error $E$ with $\wt\!\left(E\right)\leq d-1$, and can correct any erasure of $d-1$ or fewer physical qubits.
We say a quantum code $\mathcal{C}$ has parameters $\left(\!\left(n,K,d\right)\!\right)_2$ if it is a $K$-dimensional subspace of a Hilbert space on $n$ physical qubits with minimum distance $d$.
We will omit the subscript denoting the local dimension throughout the remainder of this paper, as we have restricted ourselves to qubit codes.

An important class of quantum codes are the stabilizer codes~\cite{Got1997,CRSS1998}, defined in terms of the $n$-qubit Pauli group $\mathcal{P}_n$.
Let $\mathcal{S}$ be an abelian subgroup of $\mathcal{P}_n$ not containing $-I$.
The stabilizer code $\mathcal{C}$ is the subspace stabilized by $\mathcal{S}$ (which we refer to as the stabilizer group of the code):
\begin{equation}
    \mathcal{C}=\Span\!\left\{\ket{\psi}\mid S\ket{\psi}=\ket{\psi},\;\forall S\in\mathcal{S}\right\}.
\end{equation}
If $\mathcal{S}$ has $n-k$ independent generators, then $\mathcal{C}$ can encode $k$ logical qubits, that is $\dim\!\left(\mathcal{C}\right)=K=2^k$.
As such, we often write the parameters of a stabilizer code as $\left[\!\left[n,k,d\right]\!\right]$ to distinguish them from the parameters of nonadditive (i.e., non-stabilizer) quantum codes.

For stabilizer groups, the normalizer $N\!\left(\mathcal{S}\right)$ and centralizer $C\!\left(\mathcal{S}\right)$ in $\mathcal{P}_n$ are equivalent, containing all Paulis that commute with every element of $\mathcal{S}$.
This necessarily gives us $\mathcal{S}\leq N\!\left(\mathcal{S}\right)$ as $\mathcal{S}$ is abelian.
Elements of this group act as the logical Pauli operators on the codespace.
For codes in the stabilizer framework, we can extract the minimum distance of the code as
\begin{equation}
    d=\min\wt\!\left(N\!\left(\mathcal{S}\right)-\left\langle \mathcal{S},iI\right\rangle\right),
\end{equation}
that is, $d$ is the weight of the smallest nontrivial logical operator.

\subsection{Erasure and Replacer Channels}

A quantum channel is a completely positive trace-preserving map $\mathcal{E}:\mathcal{L\!\left(H\right)} \rightarrow\mathcal{L\!\left(K\right)}$ between the operators of two Hilbert spaces $\mathcal{H}$ and $\mathcal{K}$.
For every subset of qubits $M$ of $\mathcal{H}=\overline{B}\otimes B$, there exists a family of quantum replacer channels on $\mathcal{H}$ defined by
\begin{equation}
    \mathcal{E}_B:\mathcal{L\!\left(H\right)}\rightarrow\mathcal{L\!\left(H\right)},\;\rho\mapsto \rho_{\overline{B}}\otimes\sigma_B,
\end{equation}
where $\rho_{\overline{B}}=\Tr_B\!\left(\rho\right)$ is the reduced state on $\overline{B}$ and $\sigma_B$ is a fixed ``replacer state'' that is the same for all inputs.
The replacer channel may be viewed as an erasure channel on subset $B$, followed by a reset of the lost qubits to a known state $\sigma_B$.

Recently Chitambar et al. \cite{CHKNN2025} introduced a Structure Theorem for quantum replacer codes, synthesizing past results on unitarily recoverable codes~\cite{NS2006,KS2006}.
These codes are able to recover from a quantum replacer channel $\mathcal{E}_B$ on an $b$-qubit subset $B$, that maps all states on $B$ to some fixed state $\sigma_B$.
In what follows, we restrict ourselves to error channels that are completely depolarizing on $B$, which is equivalent to a replacer channel with $\sigma_B=I_B/2^b$, the maximally-mixed state on $B$.
It is well known that the completely depolarizing channel on a set of qubits can be modeled by a quantum channel whose Kraus operators are the Pauli operators on the space; furthermore, the same result holds as long as the Kraus operators form an orthonormal basis for the operators~\cite[Lemma 4.1]{JW2022}:
\begin{equation}
    \frac{1}{\dim\!\left(\mathcal{H}\right)^2}\sum_{E\in\mathcal{E}}E\rho E^\dagger=\frac{I_{\mathcal{H}}}{\dim\!\left(\mathcal{H}\right)},
\end{equation}
for any density matrix $\rho\in\mathcal{L\!\left(H\right)}$, where $\mathcal{E}$ forms an orthonormal basis of $\mathcal{L\!\left(H\right)}$.
We say that a known subset of qubits $B$ is correctable if we can recover from a completely depolarizing channel on $B$, or equivalently, the set of correctable errors $\mathcal{E}$ of the code contains an orthonormal basis of $\mathcal{L}\!\left(B\right)$.

We give only those parts of the Structure Theorem that are relevant to this paper:
\begin{theorem}[{\cite[Structure Theorem]{CHKNN2025}}]\label{thm:structure}
    Let $\mathcal{C}$ be a quantum code with orthonormal basis $\left\{\ket{\tilde{i}}\right\}$, and let $B$ be a known subset of $b$ physical qubits. Then the following are equivalent:
    \begin{enumerate}
        \item The set of qubits $B$ is correctable for the code $\mathcal{C}$.
        \item There is a reference system $R=\Span\!\left\{\ket{i}\right\}\cong \mathcal{C}$, an ancilla system $A$ with $\left(\dim A\right)\left(\dim \mathcal{C}\right)\leq \dim\overline{B}$, an isometry $U_{\overline{B}}:RA\rightarrow \overline{B}$, and a density operator $\sigma_{AB}$, such that for all $\tilde{\rho}=\sum_{ij}\alpha_{ij}\ket{\tilde{i}}\!\bra{\tilde{j}}$ supported on $\mathcal{C}$, we have
        \begin{equation*}
            \mathcal{E}_B\!\left(\tilde{\rho}\right)= \left(\mathcal{U}_{\overline{B}}\otimes I_B\right) \left(\rho_R\otimes \sigma_{AB}\right),
        \end{equation*}
        where $\rho_R=\sum_{ij}\alpha_{ij}\ket{i}\!\bra{j}_R$ and $\mathcal{U}_{\overline{B}}\!\left(\cdot\right)=U_{\overline{B}}\!\left(\cdot\right)\! U_{\overline{B}}^\dagger$. Further, we have
        \begin{equation*}
            \sigma_{AB}=\Gamma_A\otimes I_B/2^b,
        \end{equation*}
        for some density operator $\Gamma_A$.
        \item There is a reference system $R=\Span\!\left\{\ket{i}\right\}\cong \mathcal{C}$, an ancilla system $A$ with $\left(\dim A\right)\left(\dim \mathcal{C}\right)\leq \dim\overline{B}$, an isometry $U_{\overline{B}}:RA\rightarrow \overline{B}$, and a state $\ket{\psi}_{AB}$, all determined by $B$ and $\mathcal{C}$, such that for all $\ket{\tilde{i}}\in \mathcal{C}$, we have
        \begin{equation*}
            \ket{\tilde{i}}=\left(U_{\overline{B}}\otimes I_B\right)\left(\ket{i}_R\otimes \ket{\psi}_{AB}\right).
        \end{equation*}
    \end{enumerate}

\end{theorem}

A useful characterization of correctable sets for stabilizer codes is given by the Cleaning Lemma of Bravyi and Terhal~\cite{BT2009}, which states that for a set of qubits $B$ of a stabilizer code with stabilizer group $\mathcal{S}$, then either
\begin{enumerate}
    \item there exists a non-trivial logical operator $L\in N\!\left(\mathcal{S}\right)-\left\langle\mathcal{S},iI\right\rangle$ such that $\supp\!\left(L\right)\subseteq B$, or
    \item for any logical operator $L\in N\!\left(\mathcal{S}\right)$, there exists a stabilizer $S\in\mathcal{S}$ such that $\supp\!\left(LS\right)\subseteq\overline{B}$.
\end{enumerate}
The second case, referred to as ``cleaning-off'' subsystem $B$ is an equivalent condition for $B$ being correctable.\footnote{This appears to have been a folklore result, with the first proof in the literature we were able to find having been given relatively recently by Ashikhmin in~\cite[Theorem 1]{Ash2020}. However, we note that the result follows directly from an early result of Rains on the unitary enumerators of quantum codes (see the discussion following~\cite[Theorem 10]{Rai1998}).}
As a corollary of Theorem~\ref{thm:structure}, we recover a generalization of Case (2) of the Cleaning Lemma, albeit in one direction:

\begin{corollary}\label{thm:genclean}
    Suppose $\mathcal{C}$ is a quantum code with projector $P$. If $B$ is a correctable subset of $\mathcal{C}$, then there are no non-trivial logical operators $L$ on the codespace such that $\supp\!\left(L\right)\subseteq B$.
    Moreover, any logical unitary $V$ acting on the reference system $R$ has an equivalent physical unitary operator $\widetilde{V}$ acting on $\mathcal{C}$ such that $\supp\!\left(\widetilde{V}\right)\subseteq \overline{B}$.
\end{corollary}
\begin{proof}
    The first statement follows from \cite[Corollary 4.8]{CHKNN2025}.
    The second statement can be obtained from Theorem~\ref{thm:structure} by setting $\widetilde{V}=\mathcal{U}_{\overline{B}}\!\left(V_R\otimes I_A\right)$.
\end{proof}

\subsection{Entanglement-Assisted Quantum Error Correction}

Quantum codes are designed with only the assumption that the sender and receiver share a quantum channel between them.
However, if we allow for the scenario where there is some amount of shared entanglement between the two parties, we get the more general class of entanglement-assisted (EA) quantum codes, first introduced by Bowen~\cite{Bow2002} and integrated into the stabilizer framework by Brun et al.~\cite{BDH2006,BDH2014}.
Here we review the construction of EA stabilizer codes.

As in the stabilizer case, we take a Pauli subgroup $\mathcal{S}\leq\mathcal{P}_n$ not containing $-I$ to be our stabilizer group, but here we allow $\mathcal{S}$ to be nonabelian.
We can always find a set of generators $\mathcal{S}=\langle\widetilde{Z}_1,\dots,\widetilde{Z}_{c+s},\widetilde{X}_1,\dots,\widetilde{X}_c\rangle$ that satisfy the standard Pauli commutation relations:
\begin{align}
    \left[\widetilde{Z}_i,\widetilde{Z}_j\right]=0 \quad & \quad \forall\;i,j, \\
    \left[\widetilde{X}_i,\widetilde{X}_j\right]=0 \quad & \quad \forall\;i,j, \\
    \left[\widetilde{X}_i,\widetilde{Z}_j\right]=0 \quad & \quad \forall\;i\neq j, \\
    \left\{\widetilde{X}_i,\widetilde{Z}_i\right\}=0 \quad & \quad \forall\;i,
\end{align}
where $\left[M,N\right]$ is the commutator and $\left\{M,N\right\}$ is the anticommutator of two matrices $M$ and $N$.
The isotropic subgroup is given by $\mathcal{S}_I=\langle\widetilde{Z}_{c+1},\dots,\widetilde{Z}_{c+s}\rangle$, and the symplectic (or entanglement) subgroup is given by $\mathcal{S}_S=\langle\widetilde{Z}_1,\dots,\widetilde{Z}_{c},\widetilde{X}_1,\dots,\widetilde{X}_c\rangle$.
The values of $s$ and $c$ give the number of ancilla qubits and the number of ebits required by the code.
We write the parameters of an EA stabilizer code $\left[\!\left[n,k,d;c\right]\!\right]$.

As the stabilizer group $\mathcal{S}$ is nonabelian, it does not stabilize a subspace.
However, there exists a unitary mapping $U$ such that $U\widetilde{Z}_iU^\dagger=Z_i$ (that is, a Pauli-$Z$ operator acting only on the $i$-th qubit) and $U\widetilde{X}_iU^\dagger=X_i$, giving us $c$ anticommuting $X$ and $Z$ pairs, as well as $s$ additional unpaired $Z$ operators.
To get a set of commuting operators, we can extend the symplectic subgroup generators by appending an $Z$ (respectively $X$) operator to the end of each $Z_i$ (respectively $X_i$) in an anticommuting pair, giving us a new pair of commuting stabilizer generators.
Note that a new physical qubit must be added for each anticommuting pair.
We can then apply $U^\dagger$ to the original qubits, giving an EA code.
The encoded vectors in the resulting code will have the form
\begin{equation}
    \left(U_{\overline{B}}^\dagger\otimes I_B\right)\left(\ket{\psi}_R\ket{0}_{S}^{\otimes s}\ket{\Phi}_{AB}^{\otimes c}\right),
\end{equation}
where $U_{\overline{B}}^\dagger:RSA\rightarrow\overline{B}$ and $\ket{\Phi}_{AB}=\frac{1}{\sqrt{2}}\left(\ket{00}+\ket{11}\right)$ is the maximally entangled Bell state, which we will refer to as an ebit.
If the sender and receiver share $c$ ebits, then the sender can encode the state $\ket{\psi}$ into the code with the help of $s$ ancilla states and her share of the ebits.

\begin{example}\label{ex:stabea5}

    Here we show how the original $\left[\!\left[3,1,3;2\right]\!\right]$ EA code of Bowen~\cite{Bow2002} can be cast into the stabilizer framework.
    We start with a nonabelian stabilizer group $\mathcal{S}=\langle XZZ,ZYY,ZZX,YYZ\rangle$, which can be partitioned into two anticommuting pairs.
    From here, there exists a unitary $U$ that takes the stabilizer generators to single-qubit Paulis, which can then be extended.
    These two pairs of extended generators stabilize two ebits, which can be distributed between the sender and the receiver.
    From there, the sender can apply the encoding unitary $U^\dagger$ to her message and halves of the ebits, to encode into the codespace:

    \begin{align*}
        \begin{array}{ccccc}
            \widetilde{X}_1 & = & X & Z & Z \\
            \widetilde{Z}_1 & = & Z & Y & Y \\
            \widetilde{X}_2 & = & Z & Z & X \\
            \widetilde{Z}_2 & = & Y & Y & Z \\
        \end{array}
        \overset{U}{\longmapsto} & 
        \begin{array}{ccc}
            X & I & I \\
            Z & I & I \\
            I & X & I \\
            I & Z & I \\
        \end{array} \\
        \overset{extend}{\longmapsto} & 
        \begin{array}{ccc|cc}
            X & I & I & X & I \\
            Z & I & I & Z & I \\
            I & X & I & I & X \\
            I & Z & I & I & Z \\
        \end{array} \\
        \overset{U^\dagger\otimes I_4}{\longmapsto} & 
        \begin{array}{ccc|cc}
            X & Z & Z & X & I \\
            Z & Y & Y & Z & I \\
            Z & Z & X & I & X \\
            Y & Y & Z & I & Z \\
        \end{array}
    \end{align*}

    These extended stabilizers stabilize a codespace that is identical to that of the perfect $\left[\!\left[5,1,3\right]\!\right]$ stabilizer code~\cite{LMPZ1996,BDSW1996}, but spread across the sender and receiver.
\end{example}

If we assume that the pre-shared ebits are noiseless, then any noise in the system must occur on those qubits transferred over the quantum channel between the sender to the receiver. Similar to the non-EA case, we can calculate the minimum distance of an EA stabilizer code from the stabilizer group:
\begin{equation}
    d=\min\wt\!\left(N\!\left(\mathcal{S}\right)-\left\langle \mathcal{S}_I,iI\right\rangle\right),
\end{equation}
that is, $d$ is the weight of the smallest nontrivial logical operator on sender's qubits.

Of primary interest to this paper is the question of how to construct EA codes from preexisting quantum codes.
While there have been multiple results along these lines (c.f.~\cite{LB2012,GHMR2019,GHMR2021,GHW2022,UM2022}), these results are primarily restricted to nondegenerate stabilizer codes and are done by puncturing process on subsets of size less than $d$.
The most complete is the following result from Grassl et al.~\cite{GHW2022}:
\begin{theorem}[{\cite[Theorem 10]{GHW2022}}]
    Let $\mathcal{C}$ be a pure $\left[\!\left[n,k,d\right]\!\right]$ quantum code. Then an entanglement-assisted code with parameters $\left[\!\left[n-c,k,d;c\right]\!\right]$ exists for all $c<d$.
\end{theorem}

\begin{remark}\label{thm:remarkGHW}
    The construction given by Grassl et al. in~\cite[Theorem 10]{GHW2022} distinguishes itself from the other results on the subject in that it applies to a general pure quantum code, including those outside the stabilizer framework.
    However, this may have been overlooked due to their use of the stabilizer-style parameters for the resulting EA code, as well as the fact that they do not give any examples of nonadditive EA codes.

    Later in this work we remedy this by providing the first examples of EA codes outside of the stabilizer and codeword-stabilized frameworks in Examples~\ref{ex:eapi4}, \ref{ex:eaxp7}, and~\ref{ex:eapi7}.
    Moreover, we generalize their result to any quantum code in Theorems~\ref{thm:eauncondensed} and~\ref{thm:eacondensed}, relying on the more general notion of correctable subsets, which subsumes sets of size less than the minimum distance.
\end{remark}

\section{Entanglement-Assisted Subspace Codes}

When constructing EA codes, there are two error models possible: noiseless ebits~\cite{BDH2006} and noisy ebits~\cite{LB2012}.
In the noiseless ebit error model, we assume the ebits are not affected by any noise, so that any errors must occur on the qubits of $\overline{B}$ as they are transmitted across the quantum channel.
In the noisy ebit error model, we do not make this assumption, and allow for imperfections in the ebits prior to the start, as well as errors that may occur on $B$ while waiting for $\overline{B}$ to be transmitted.

As previous research on EA codes has primarily been restricted to the stabilizer framework, the EA error correction conditions were also formulated in the same language.
We start our investigation into nonadditive EA codes by giving Knill-Laflamme style error correction conditions.

\begin{theorem}[Error Correction Condition for Noiseless Ebits]
    Let $\mathcal{C}\subseteq\overline{B}\otimes B$ be a quantum code with projector $P$. Then $\mathcal{C}$ corrects the set of errors $\mathcal{E}\subseteq\mathcal{L}\!\left(\overline{B}\right)$ if and only if 
    \begin{equation}
        P\left(E_i^\dagger E_j\otimes I_B\right)P=\lambda_{ij}P,
    \end{equation}
    for all $E_i,E_j\in\mathcal{E}$, where $\lambda_{ij}\in\mathbb{C}$.
\end{theorem}

\begin{theorem}[Error Correction Condition for Noisy Ebits]
    Let $\mathcal{C}\subseteq\overline{B}\otimes B$ be a quantum code with projector $P$. Then $\mathcal{C}$ corrects the set of errors $\mathcal{E}\subseteq\mathcal{L}\!\left(\overline{B}\otimes B\right)$ if and only if 
    \begin{equation}
        P\left(E_i^\dagger E_j\right)P=\lambda_{ij}P,
    \end{equation}
    for all $E_i,E_j\in\mathcal{E}$, where $\lambda_{ij}\in\mathbb{C}$.
\end{theorem}

Note that the former result is a special case of the standard error correction condition in Theorem~\ref{thm:klcond}, while the latter is equivalent to it.
The results in the remainder of this section will typically be valid for both the noiseless and noisy ebit error models, while those in Section~\ref{sec:reduceentang} will typically apply only to the noiseless ebit error model.

The key insight in constructing EA codes from quantum codes is in identifying which subsets the receiver is allowed to have before the transmission of the information.
Corollary~\ref{thm:genclean}, a partial generalization of the Cleaning Lemma~\cite{BT2009}, says that if a set of qubits $B$ is correctable, i.e., the code can recover from an erasure of $B$, then no non-trivial logical operator has support on $B$; additionally, $\overline{B}$ has a full set of logical operators supported on it.
Essentially, it is not necessary for the sender to have access to $B$ at the time of encoding, as she can always construct a unitary $\widetilde{V}$ acting only on $\overline{B}$ that maps $\ket{\widetilde{0}}\mapsto\ket{\widetilde{\psi}}$, where $\ket{\widetilde{\psi}}$ is the encoded state she intends to send.

\begin{theorem}\label{thm:easendahead}
    Suppose $\mathcal{C}$ is an $\left(\!\left(n,K,d\right)\!\right)$ quantum code with a correctable set $B$ of $b$ qubits.
    Then by pre-sending $B$ to the receiver, $\mathcal{C}$ is also an $\left(\!\left(n-b,K,d;2^b\right)\!\right)$ entanglement-assisted code.
\end{theorem}
\begin{proof}
    Follows from Corollary~\ref{thm:genclean} and the prior discussion.
\end{proof}

Here, we denote by $C$ the dimension of the subsystem held by the receiver, and write the parameters of a general EA code as $\left(\!\left(n,K,d;C\right)\!\right)$ in analogy to the parameters of nonadditive codes.
This allows for bipartite states beyond the $m$ shared Bell pairs in the EA stabilizer case, including for bipartite states that are not maximally entangled between the sender and receiver.
In the case of an $\left[\!\left[n,k,d;c\right]\!\right]$ EA stabilizer code, we have $C=2^c$.

An alternative strategy is suggested by Theorem~\ref{thm:structure}, the Structure Theorem of~\cite{CHKNN2025}.
For a density operator $\widetilde{\rho}$ on the codespace, the reduced state after tracing out a correctable subset $B$ has the form
\begin{equation}
    \Tr_B\!\left(\widetilde{\rho}\right)=\widetilde{\rho}_{\overline{B}} =U_{\overline{B}}\left(\rho_R\otimes\Gamma_A\right)U_{\overline{B}}^\dagger,
\end{equation}
where $U_{\overline{B}}:RA\rightarrow \overline{B}$ is an isometry and $\Gamma_A$ is some density matrix.
Then there is a purification $\ket{\psi}_{AB}$ of $\Gamma_A$ such that
\begin{equation}
    \ket{\widetilde{i}}=\left(U_{\overline{B}}\otimes I_B\right)\left(\ket{i}_R\otimes\ket{\psi}_{AB}\right)
\end{equation}
for all basis states $\ket{\widetilde{i}}$ of $\mathcal{C}$.

\begin{theorem}\label{thm:eauncondensed}
    Suppose $\mathcal{C}$ is an $\left(\!\left(n,K,d\right)\!\right)$ quantum code with a correctable set $B$ of $b$ qubits.
    Then there exists a bipartite state $\ket{\psi}_{AB}$ and encoding isometry $U_{\overline{B}}$ the sender and receiver can use that makes $\mathcal{C}$ an $\left(\!\left(n-b,K,d;2^b\right)\!\right)$ entanglement-assisted code.
\end{theorem}
\begin{proof}
    Follows from Theorem~\ref{thm:structure} and the prior discussion.
\end{proof}

Any EA code constructed using either of these strategies is identical to the original code, and is therefore able to correct the same set of errors.
For the remainder of the paper, we focus on codes constructed using the latter approach, as they generalize the well-established EA schemes given in~\cite{Bow2002,BDH2006,LB2012}.
Below we show how the Structure Theorem can be applied to produces some old and new EA codes.

\begin{example}
    In \cite[Example 4.7]{CHKNN2025}, the authors applied the Structure Theorem to the perfect $\left[\!\left[5,1,3\right]\!\right]$ stabilizer code, obtaining an encoding unitary and entangled state comprised of two ebits. These are identical to the encoding unitary and shared entanglement of the original $\left[\!\left[3,1,3;2\right]\!\right]$ EA code constructed by Bowen~\cite{Bow2002} from the $\left[\!\left[5,1,3\right]\!\right]$ code (see Example~\ref{ex:stabea5}).
\end{example}

To date, all examples of EA codes have been given either in the stabilizer formalism~\cite{BDH2006}, its generalization the codeword-stabilized (CWS) formalism~\cite{SHB2011,AJH2016}, or teleportation-style schemes~\cite{Gra2021}.
Theorem~\ref{thm:eauncondensed} allows for the construction of EA codes from any quantum code, including those outside of the stabilizer and codeword-stabilized frameworks.
As mentioned in Remark~\ref{thm:remarkGHW}, this construction generalizes the one given by Grassl et al. in~\cite[Theorem 10]{GHW2022}, which also works for pure codes outside the stabilizer framework.
However, our construction also works for impure codes (which is expanded on in Section~\ref{sec:reduceentang}), as well as for any correctable set rather than just those of size less than the minimum distance $d$.
Below we give examples of EA codes from permutation-invariant~\cite{PR2004} and XP stabilizer codes~\cite{NBV2015,WBB2022}, two families outside the normal stabilizer and CWS frameworks.

\begin{example}\label{ex:eapi4}
    Permutation-Invariant (PI) codes are invariant under an arbitrary permutation of the physical qubits. As such, they are subspaces of the symmetric subspace of $n$-qubits, which has as a basis the Dicke states $\left\{\ket{D_i^n}\right\}_{i=0}^n$, where
    \begin{equation*}
        \ket{D_i^n}=\frac{1}{\sqrt{\binom{n}{i}}}\sum\limits_{\wt\left(x\right)=i}\ket{x}
    \end{equation*}
    is a uniform superposition of the computational basis states with Hamming weight~$i$. Aydin et al.~\cite{AAB2024} defined a $\left(\!\left(4,2,2\right)\!\right)$ PI code with logical basis states
    \begin{align*}
        \ket{\tilde{0}} & = \frac{\sqrt{3}}{3}\ket{D_0^4}+\frac{\sqrt{6}}{3}\ket{D_3^4}, \\
        \ket{\tilde{1}} & = \frac{\sqrt{6}}{3}\ket{D_1^4}-\frac{\sqrt{3}}{3}\ket{D_4^4}.
    \end{align*}

    We follow Theorem~\ref{thm:structure} for an erasure of one physical qubit (WLOG we choose qubit 4). Choosing $\tilde{\rho}=\ket{\tilde{0}}\!\bra{\tilde{0}}$, we see that
    \begin{equation*}
        \Tr_{4}\!\left(\tilde{\rho}\right)=\frac{1}{2}\left(\ket{c_0}\!\bra{c_0}+\ket{c_1}\!\bra{c_1} \right),
    \end{equation*}
    where
    \begin{align*}
        \ket{c_0} & = \frac{\sqrt{6}}{3}\ket{D_0^3}+\frac{\sqrt{3}}{3}\ket{D_3^3}, \\
        \ket{c_1} & = \ket{D_2^3}.
    \end{align*}
    Therefore $A=\mathbb{C}^2$ and $\Gamma_A=\frac{1}{2}I_2$, so we can define the isometry $U_{\overline{B}}$ via its action on basis states as:
    \begin{equation*}
        \begin{array}{cccc} 
            R & A & U_{\overline{B}}  & \overline{B}  \\ \hline 
            \ket{0} & \left\{ \begin{array}{l}  \ket{0} \\ \ket{1} \end{array} \right. & \longmapsto & \left\{ \begin{array}{l} \frac{\sqrt{6}}{3}\ket{D_0^3}+\frac{\sqrt{3}}{3}\ket{D_3^3} \\ \ket{D_2^3}  \end{array} \right. \\ 
            \ket{1} & \left\{ \begin{array}{l}  \ket{0} \\ \ket{1} \end{array} \right. & \longmapsto & \left\{ \begin{array}{l} \ket{D_1^3} \\ \frac{\sqrt{3}}{3}\ket{D_0^3}-\frac{\sqrt{6}}{3}\ket{D_3^3}  \end{array} \right. \\  
        \end{array}
    \end{equation*}
    Then for $i=0,1$ we have
    \begin{equation*}
        \ket{\tilde{i}}=\left(U_{\overline{B}}\otimes I_B\right)\left(\ket{i}_R\otimes\ket{\psi}_{AB}\right),
    \end{equation*}
    where $\ket{\psi}_{AB}=\frac{1}{\sqrt{2}}\left(\ket{00}+\ket{11}\right)$ is the purification of $\Gamma_A$.
    
    Assuming the sender and receiver hold subsystems $A$ and $B$ respectively, the sender is able to encode the message along with her half of the ebit using the isometry $U_{\overline{B}}$.
    Therefore, we have a $\left(\!\left(3,2,2;2\right)\!\right)$ EA PI code.
\end{example}

%%% actual XP example

\begin{example}\label{ex:eaxp7}
The XP stabilizer formalism is a generalization of the standard stabilizer formalism in that the stabilizer generators are taken from the group $\left\langle \omega I,X,P\right\rangle^{\otimes n}$, where $\omega=e^{i\pi/N}$, $P=\diag\!\left(1,\omega^2\right)$, and $N$ is a fixed positive integer known as the precision. Webster et al.~\cite[Example 6.3]{WBB2022} gave a $\left(\!\left(7,8,2\right)\!\right)$ XP stabilizer code with precision $N=8$, with encoded basis states
\begin{align*}
    \ket{\widetilde{0}} & = \frac{1}{\sqrt{2}}\left(\ket{0000000}+\omega^{12}\ket{1111111}\right) \\
    \ket{\widetilde{1}} & = \frac{1}{\sqrt{2}}\left(\ket{0000111}+\ket{1111000}\right) \\
    \ket{\widetilde{2}} & = \frac{1}{\sqrt{2}}\left(\ket{0001011}+\omega^{14}\ket{1110100}\right) \\
    \ket{\widetilde{3}} & = \frac{1}{\sqrt{2}}\left(\ket{0001101}+\omega^{12}\ket{1110010}\right) \\
    \ket{\widetilde{4}} & = \frac{1}{\sqrt{2}}\left(\ket{0011110}+\ket{1100001}\right) \\
    \ket{\widetilde{5}} & = \frac{1}{\sqrt{2}}\left(\ket{0011001}+\omega^{8}\ket{1100110}\right) \\
    \ket{\widetilde{6}} & = \frac{1}{\sqrt{2}}\left(\ket{0010101}+\omega^{10}\ket{1101010}\right) \\
    \ket{\widetilde{7}} & = \frac{1}{\sqrt{2}}\left(\ket{0010011}+\omega^{12}\ket{1101100}\right) 
\end{align*}

Let the correctable set $B$ be qubit 7. Following Theorem~\ref{thm:structure}, we take the logical state $\ket{\widetilde{0}}$ and compute the reduced density matrix tracing out $B$:
\begin{equation}
    \rho_B = \Tr_{\overline{B}}\!\left(\ket{\widetilde{0}} \bra{\widetilde{0}}\right) = \frac{1}{2}(\ket{000000}\bra{000000} + \ket{111111}\bra{111111}),
\end{equation}
meaning we can take $\Gamma_A = \frac{1}{2} I_2$ and $A \cong \mathbb{C}^2$.
We can construct an isometry $U_{\overline{B}}: R A \to \overline{B}$ such that
\begin{equation}
    |\tilde{i}\rangle = (U_{\overline{B}} \otimes I_B)(|i\rangle_R \otimes |\psi\rangle_{AB}),
\end{equation}

where $|\psi\rangle_{AB} = \frac{1}{\sqrt2}(|00\rangle + |11\rangle)$ is the purification of $\Gamma_A$. Explicitly, $U_{\overline{B}}$ acts as:

\[
\begin{array}{cccc} 
R & A & U_{\overline{B}}  & \overline{B}  \\ \hline 
\ket{0} & \left\{ \begin{array}{l}  \ket{0} \\ \ket{1} \end{array} \right. & \longmapsto & \left\{ \begin{array}{l} \ket{000000} \\ \omega^{12}\ket{111111} \end{array} \right. \\ 
\ket{1} & \left\{ \begin{array}{l}  \ket{0} \\ \ket{1} \end{array} \right. & \longmapsto & \left\{ \begin{array}{l} \ket{111100} \\ \ket{000011} \end{array} \right. \\
\ket{2} & \left\{ \begin{array}{l}  \ket{0} \\ \ket{1} \end{array} \right. & \longmapsto & \left\{ \begin{array}{l} \omega^{14}\ket{111010} \\ \ket{000101} \end{array} \right. \\  
\ket{3} & \left\{ \begin{array}{l}  \ket{0} \\ \ket{1} \end{array} \right. & \longmapsto & \left\{ \begin{array}{l} \omega^{12}\ket{111001} \\ \ket{000110} \end{array} \right. \\  
\ket{4} & \left\{ \begin{array}{l}  \ket{0} \\ \ket{1} \end{array} \right. & \longmapsto & \left\{ \begin{array}{l} \ket{001111} \\ \ket{110000} \end{array} \right. \\  
\ket{5} & \left\{ \begin{array}{l}  \ket{0} \\ \ket{1} \end{array} \right. & \longmapsto & \left\{ \begin{array}{l} \omega^8\ket{110011} \\ \ket{001100} \end{array} \right. \\  
\ket{6} & \left\{ \begin{array}{l}  \ket{0} \\ \ket{1} \end{array} \right. & \longmapsto & \left\{ \begin{array}{l} \omega^{10}\ket{110101} \\ \ket{001010} \end{array} \right. \\ 
\ket{7} & \left\{ \begin{array}{l}  \ket{0} \\ \ket{1} \end{array} \right. & \longmapsto & \left\{ \begin{array}{l} \omega^{12}\ket{001001} \\ \ket{110110} \end{array} \right. \\ 
\end{array}
\]

Assuming the sender holds $A$ and the receiver holds $B$, the sender encodes using $U_{\overline{B}}$ and transmits the 6 qubits of $\overline{B}$ to the receiver, we obtain a $(\!(6,8,2;2)\!)$ EA XP code.
\end{example} 

\section{Receiver Share Compression via Degeneracy}\label{sec:reduceentang}

The  EA codes constructed using Theorem~\ref{thm:eauncondensed} are valid for both the noisy and noiseless ebit error models.
However, specifically for the noiseless ebit error model, this extra protection on the receiver's qubits may be seen as unnecessary and potentially wasteful if it requires more qubits than necessary.
A natural question to ask is whether the size of the receiver's share can be reduced.
In this section, we show that it is only possible to reduce the share size held by the receiver if the code in question is degenerate for the erasure channel on the correctable subset $B$.

In order to classify the density matrix $\Gamma_A$ on the ancilla system and its purifications, we will need two dichotomies from the coding literature: that between degeneracy and nondegeneracy, the other between purity and impurity.
Let $\mathcal{C}$ be a quantum code that can correct a linearly-independent error set $\mathcal{E}$.
By the Knill-Laflamme conditions, we have $PE_i^\dagger E_jP=\lambda_{ij}P$, for all $E_i,E_j\in \mathcal{E}$.
It is well known that the matrix $\lambda$ formed by the coefficients is Hermitian and therefore diagonalizable.
We say that the code is nondegenerate if $\rank\!\left(\lambda\right)=\left\lvert\mathcal{E}\right\rvert$, and degenerate if $\rank\!\left(\lambda\right)<\left\lvert\mathcal{E}\right\rvert$.

A code is pure if $PEP=0$ for all traceless errors $E\in\mathcal{E}$, and is impure otherwise.
For a replacer channel affecting $B$, this is equivalent to every encoded state in $\mathcal{C}$ having a maximally mixed reduced density operator of rank $2^b$ when $B$ is traced out.
For an arbitrary quantum code, purity implies nondegeneracy, giving the following trichotomy:

\begin{theorem}\label{thm:degenentanglement}
    Let $\mathcal{C}$ be a quantum code and $B$ a correctable subset of $b$ qubits. Then one of the following is true:
    \begin{enumerate}
        \item $\mathcal{C}$ is pure with respect to the erasure of $B$ and $\rho_{\overline{B}}$ is maximally mixed with rank $2^b$.
        \item $\mathcal{C}$ is impure and nondegenerate with respect to the erasure of $B$ and $\rho_{\overline{B}}$ is not maximally mixed but has rank $2^b$.
        \item $\mathcal{C}$ is degenerate with respect to the erasure of $B$ and $\rho_{\overline{B}}$ has rank less than $2^b$.
    \end{enumerate}
\end{theorem}
\begin{proof}

For Case (1), it is well known that the code being pure is equivalent to having a maximally mixed marginal for the respective erasure (see ~\cite{GHW2022}).

For Case (3), we first diagonalize the Knill-Laflamme matrix $\lambda$ by finding an orthonormal basis of error operators $\left\{F_i\right\}$ such that $\lambda_{ij}=0$ if $i\neq j$.
Note that if $\varrho=P/K$ is the maximally mixed state of the codespace, then
\begin{equation}
    \lambda_{ij}=\Tr\!\left(\varrho \left(I\otimes F_i^\dagger\right)\left(I\otimes F_j\right)\right)=\Tr\!\left(\varrho_B F_i^\dagger F_j\right),
\end{equation}
and for the diagonal terms we have
\begin{equation}
    \lambda_{ii}=\norm{\varrho_B^{1/2}F_i^\dagger}^2.
\end{equation}

Suppose that $\varrho_B$ is full rank.
Then $\varrho_B^{1/2}$ is invertible and $\varrho_B^{1/2}F_i^\dagger=0$ implies $F_i^\dagger=0$.
However, this can't happen because $\left\{F_i\right\}$ was chosen to be an orthonormal basis.
Therefore $\lambda$ is also full rank.

Now suppose that $\rank\!\left(\varrho_B\right)=r<2^b$.
We can choose our basis so that
\begin{equation}
    \varrho_B^{1/2}=\begin{pmatrix}
        \varrho_B^{1/2} & 0 \\
        0 & 0
    \end{pmatrix},
\end{equation}
then any error of the form
\begin{equation}
    F_i=\begin{pmatrix}
        0 & 0 \\
        F & 0
    \end{pmatrix}
\end{equation}
satisfies $\varrho_B^{1/2}F_i^\dagger=0$, hence $\lambda_{ii}=0$. Since there are $r\left(2^b-r\right)$ such linearly independent operators, it follows that $\lambda$ is not full rank.

Case (2) contains all scenarios not already covered.
\end{proof}

For the special case of stabilizer codes, purity and nondegeneracy are equivalent, so Case (2) cannot occur.
This gives a nice characterization of pure/nondegenerate stabilizer codes in terms of the shared bipartite state:

\begin{corollary}
    Let $\mathcal{C}$ be a stabilizer code and $B$ a correctable subset of $b$ qubits. Then we can choose $\ket{\psi}_{AB}$ to be $b$ ebits if and only if $\mathcal{C}$ has no stabilizers with support on $B$.
\end{corollary}

Theorem~\ref{thm:degenentanglement} completely determines when the receiver's share $B$ can be compressed: precisely when the original code is degenerate for the erasure of $B$.
Using this knowledge, we can modify Theorem~\ref{thm:eauncondensed} to allow for a smaller receiver share size.

\begin{theorem}\label{thm:eacondensed}
    Suppose $\mathcal{C}$ is an $\left(\!\left(n,K,d\right)\!\right)$ quantum code with a correctable set $B$ of $b$ qubits, and let $\Gamma_A$ have a purification $\ket{\psi'}_{AB}$ with Schmidt rank $C$.
    If the state $\ket{\psi'}_{AB}$ is divided between the sender and receiver, then there exists an $\left(\!\left(n-b,K,d;C\right)\!\right)$ entanglement-assisted code for the noiseless ebit error model.
    Additionally, $C<2^b$ if and only if the code $\mathcal{C}$ is degenerate for the erasure of $B$.
\end{theorem}
\begin{proof}
    By the Structure Theorem, there exists an isometry $U_{\overline{B}}: RA\rightarrow \overline{B}$ such that
    \begin{equation}
        \ket{\tilde{i}}=\left(U_{\overline{B}}\otimes I_B\right)\left(\ket{i}_R \otimes \ket{\psi}_{AB}\right).
    \end{equation}
    As all purifications of $\Gamma_A$ are equivalent up to an isometry on the purifying system, there exists a purification $\ket{\psi'}_{AB'}$ and isometry $V:B'\rightarrow B$ such that $\dim\!\left(B'\right)=C$ and $\ket{\psi}_{AB}=\left(I_A\otimes V_{B'}\right)\ket{\psi'}_{AB'}$.
    Therefore, upon receiving $\overline{B}$ from the sender, the receiver can apply $V$ to his share $B'$ and then recover using a standard decoder for the code $\mathcal{C}$.
\end{proof}

One special edge case is when the codespace is comprised of product states between $\overline{B}$ and $B$, so that an encoded state $\ket{\widetilde{\phi}}\in\mathcal{C}$ has the form
\begin{equation}
    \ket{\widetilde{\phi}}=\ket{\widetilde{\phi}'}_{\overline{B}}\otimes\ket{\varphi}_B,
\end{equation}
where $\ket{\varphi}_B$ is independent of the information encoded.
Codes such as these often appear in code propagation rules~\cite[Theorem 6]{CRSS1998}, and are always degenerate for an erasure of $B$.
In this case, the shared bipartite state $\ket{\psi'}$ from Theorem~\ref{thm:eacondensed} is also a product state, and the resulting EA code is now effectively a standard quantum code that has had its useless subsystem $B$ removed.

Using the construction in Theorem~\ref{thm:eacondensed}, we look at two examples where the codes are degenerate for the channels erasing their respective correctable subsets.

\begin{example}\label{ex:eapi7}
    Kubischta and Teixeira~\cite{KT2023} defined a $\left(\!\left(7,2,3\right)\!\right)$ PI code with logical basis states
    \begin{align*}
        \ket{\tilde{0}} & = \frac{\sqrt{15}}{8}\ket{D_0^7}+\frac{\sqrt{7}}{8}\ket{D_2^7}+\frac{\sqrt{21}}{8}\ket{D_4^7}-\frac{\sqrt{21}}{8}\ket{D_6^7}, \\
        \ket{\tilde{1}} & = -\frac{\sqrt{21}}{8}\ket{D_1^7}+\frac{\sqrt{21}}{8}\ket{D_3^7}+\frac{\sqrt{7}}{8}\ket{D_5^7}+\frac{\sqrt{15}}{8}\ket{D_7^7}.
    \end{align*}
    
    We follow the Structure Theorem for an erasure of two physical qubits (WLOG we choose qubits 6 and 7). Choosing $\tilde{\rho}=\ket{\tilde{0}}\!\bra{\tilde{0}}$, we see that
    \begin{equation*}
        \Tr_{6,7}\!\left(\tilde{\rho}\right)=\frac{1}{3}\left(\ket{c_0}\!\bra{c_0}+\ket{c_1}\!\bra{c_1}+\ket{c_2}\!\bra{c_2} \right),
    \end{equation*}
    where
    \begin{align*}
        \ket{c_0} & = -\frac{\sqrt{10}}{8}\ket{D_1^5}-\frac{3}{4}\ket{D_3^5}+\frac{3\sqrt{2}}{8}\ket{D_5^5}, \\
        \ket{c_1} & = \frac{\sqrt{70}}{14}\ket{D_0^5}+\frac{3\sqrt{14}}{14}\ket{D_4^5}, \\
        \ket{c_2} & = \frac{\sqrt{6}}{3}\ket{D_0^5}+\frac{\sqrt{3}}{3}\ket{D_2^5}.
    \end{align*}
    Therefore $A=\mathbb{C}^3$ and $\Gamma_A=\frac{1}{3}I_3$, so we can define the isometry $U_{\overline{B}}$ via its action on basis states as:
    \begin{equation*}
        \begin{array}{cccc} 
            R & A & U_{\overline{B}}  & \overline{B}  \\ \hline 
            \ket{0} & \left\{ \begin{array}{l}  \ket{D_0^2} \\ \ket{D_1^2} \\ \ket{D_2^2} \end{array} \right. & \longmapsto & \left\{ \begin{array}{l} \frac{3\sqrt{5}}{8}\ket{D_0^5}+\frac{\sqrt{10}}{8}\ket{D_2^5}+\frac{3}{8}\ket{D_4^5} \\ \frac{\sqrt{10}}{8}\ket{D_1^5}+\frac{3}{4}\ket{D_3^5}-\frac{3\sqrt{2}}{4}\ket{D_5^5} \\ \frac{1}{8}\ket{D_0^5}+\frac{3\sqrt{2}}{8}\ket{D_2^3}-\frac{3\sqrt{5}}{8}\ket{D_4^5} \end{array} \right. \\ 
            \ket{1} & \left\{ \begin{array}{l}  \ket{D_0^2} \\ \ket{D_1^2} \\ \ket{D_2^2} \end{array} \right. & \longmapsto & \left\{ \begin{array}{l} -\frac{3\sqrt{5}}{8}\ket{D_1^5}+\frac{3\sqrt{2}}{8}\ket{D_3^5}+\frac{1}{8}\ket{D_5^5} \\ -\frac{3\sqrt{2}}{8}\ket{D_0^5}+\frac{3}{4}\ket{D_2^5}+\frac{\sqrt{10}}{8}\ket{D_4^5} \\ \frac{3}{8}\ket{D_1^5}+\frac{\sqrt{10}}{8}\ket{D_3^5}+\frac{3\sqrt{5}}{8}\ket{D_5^5} \end{array} \right. \\  
        \end{array}
    \end{equation*}
    Then for $i=0,1$ we have
    \begin{equation*}
        \ket{\tilde{i}}=\left(U_{\overline{B}}\otimes I_B\right)\left(\ket{i}_R\otimes\ket{\psi}_{AB}\right),
    \end{equation*}
    where $\ket{\psi}_{AB}=\frac{1}{\sqrt{3}}\left(\ket{D_0^2}_A\ket{D_0^2}_B+\ket{D_1^2}_A\ket{D_1^2}_B+\ket{D_2^2}_A\ket{D_2^2}_B\right)$ is the purification of $\Gamma_A$.
    This gives a shared bipartite state not equal to two ebits, but rather a maximally entangled state between two symmetric subspaces on two qubits.
    
    Assuming the sender and receiver hold subsystems $A$ and $B$ respectively, the sender is able to encode the message along with her half of the entangled pair using the isometry $U_{\overline{B}}$.
    Therefore we have a $\left(\!\left(5,2,3;3\right)\!\right)$ EA PI code.
\end{example}

In general, any PI code of distance $d>2$ will be degenerate for erasures of between 2 and $d-1$ qubits.
Let $B$ be a correctable set with a SWAP operator supported on it, let $\mathcal{E}$ be the set of linearly correctable errors for the erasure of $B$, and let $F=\mathrm{SWAP}-I=\sum_j f_j E_j$, where $E_j\in\mathcal{E}$.
Then we have $FP=0$, since the codespace is invariant under the SWAP operator.
For any $E_i$, we have
\begin{equation}
    0=PE_i^\dagger FP=\sum_j f_j PE_i^\dagger E_jP=\left(\sum_j f_j \lambda_{ij}\right)P,
\end{equation}
causing $\sum_j f_j \lambda_{ij}=0$ for any $i$.
Therefore $\lambda$ is not full rank.
This implies that some amount of compression is always an option if the receiver is willing to give up permutation-invariance on his share.

\begin{remark}
    Note that even if the sender and receiver do not share the bipartite state $\ket{\psi}_{AB}$ with Schmidt rank $C$ required for their particular code, there is an LOCC transformation that takes $\left\lceil\log_2 C\right\rceil$ ebits shared between the sender and receiver and transforms them into $\ket{\psi}_{AB}$.
    This is due to Nielsen's majorization theorem~\cite[Theorem 1]{Nie1999}, as $\left\lceil\log_2 C\right\rceil$ maximally entangled states are majorized by any state with Schmidt rank $C$.
\end{remark}

When we restrict Theorem~\ref{thm:eacondensed} to stabilizer codes, we obtain a characterization in terms of subgroups of the stabilizer group:

\begin{corollary}\label{thm:stabcompressed}
    Let $\mathcal{C}$ be an $\left[\!\left[n,k,d\right]\!\right]$ stabilizer code. Let $B$ be a correctable set of qubits on $\mathcal{C}$, and let $\mathcal{S}_B$ be the subgroup of $\mathcal{S}$ with support on $B$, that is
    \begin{equation*}
        \mathcal{S}_B=\left\{S\in\mathcal{S}\mid \supp\!\left(S\right)\subseteq B\right\}.
    \end{equation*}
    Let $\left\lvert B\right\rvert=b$ and $\left\lvert \mathcal{S}_B\right\rvert=2^s$.
    Then there exists an entanglement-assisted code with parameters $\left[\!\left[n-b,k,d;b-s\right]\!\right]$.
\end{corollary}

\begin{example}
    The $\left[\!\left[7,1,3\right]\!\right]$ Steane code~\cite{Ste1996} has the following stabilizer generators:
    \begin{equation*}
        \begin{array}{ccccccc}
            I & I & I & X & X & X & X \\
            X & I & X & I & X & I & X \\
            I & X & X & I & I & X & X \\
            I & I & I & Z & Z & Z & Z \\
            Z & I & Z & I & Z & I & Z \\
            I & Z & Z & I & I & Z & Z \\
        \end{array}
    \end{equation*}
    The code has multiple 4-qubit subsets that are correctable, for example, $B=\left\{4,5,6,7\right\}$, all of which have a pair of stabilizer generators supported on them; therefore, the code is degenerate for the erasure of $B$.
    Since $\left\lvert\mathcal{S}_B\right\rvert=2^2$, by Corollary~\ref{thm:stabcompressed} we should be able to transform the code into a $\left[\!\left[3,1,3;2\right]\!\right]$ EA code.
    
    This can be done by picking a purification of $\Gamma_A$ that minimizes the number of qubits held by the receiver.
    Following Fattal et al.~\cite{FCYBC2004}, stabilizers generators that have partial support on $B$ can be paired up into anticommuting pairs isomorphic to $X_i$ and $Z_i$, while those fully supported on $\mathcal{S}_B$ map to other $Z_i$ operators.
    Therefore similar to Example~\ref{ex:stabea5}, we can pick a unitary $V$ acting only on $B$ that decodes only that part:
    \begin{equation*}
        \begin{array}{ccc|cccc}
            I & I & I & I & I & Z & I \\
            X & I & X & X & I & I & I \\
            I & X & X & I & X & I & I \\
            I & I & I & I & I & I & Z \\
            Z & I & Z & I & Z & I & I \\
            I & Z & Z & Z & I & I & I \\
        \end{array}
        \overset{punc}{\longmapsto}
        \begin{array}{ccc|cc}
            X & I & X & X & I \\
            I & X & X & I & X \\
            Z & I & Z & I & Z \\
            I & Z & Z & Z & I \\
        \end{array}
    \end{equation*}
    Since the last two qubits are now disentangled from the rest of the codeword, we can ``puncture'' and discard them.
\end{example}

As noted previously, this can be seen as stripping away the error correction protection on the receiver's qubit.
Therefore, it can be seen in a sense as a reversal of the scheme proposed by Lai and Brun~\cite{LB2012}, in which the receiver's unencoded qubits are encoded into another quantum code.

\section{Conclusion}

Entanglement-assisted codes form a natural extension of quantum subspace codes that allow the sender and receiver to make use of shared entanglement.
Previously, the only known EA codes were constructed via the stabilizer framework and its direct generalization, the CWS framework.
In this paper, our main result was to show that any quantum code may be turned into an EA code by having the receiver hold a correctable subset of the code ahead of time, or by choosing a shared bipartite state and local encoding unitary using the Structure Theorem of~\cite{CHKNN2025}.
Furthermore, we show a direct connection between code degeneracy and being able to reduce the receiver's share size.

Generalizing EA codes outside the stabilizer framework opens the door for other generalizations.
While we restricted ourselves in this paper to qubit codes, the fact that our main tool, the Structure Theorem of~\cite{CHKNN2025}, is not constrained by this suggests that a generalization to qudit codes should be straightforward.
A more interesting generalization is the integration of EA codes into the OAQEC framework~\cite{BKK2007mar,BKK2007oct,DKV2024}, combining shared entanglement with subsystem~\cite{KLP2005,Pou2005} and hybrid quantum-classical codes~\cite{GLZ2017,NK2018,NK2021,NK2022}.
Recent work down this path in the stabilizer framework has been done by Nadkarni et al.~\cite{NADKV2024}.

While the results in this paper are quite general, the definition of the local unitary is quite unwieldy.
A natural question to ask is whether or not specific classes of EA codes admit a more compact description, similar to the EA stabilizer codes, which may be described using the generating set of the stabilizer group.
An obvious first choice would be EA XP stabilizer codes~\cite{WBB2022}, which similarly have stabilizer groups but with a more complicated structure.
Another option would be EA PI codes, as Aydin et al.~\cite{AAB2024} have given a concise description of 2-dimensional PI codes in terms of coefficient vectors describing the superpositions of Dicke states.

These results may additionally be useful in generalizing both Grassl's teleportation scheme~\cite{Gra2021} and the bipartite stabilizer codes of Wilde and Fattal~\cite{WF2010}, as well as connect to quantum codes over networks~\cite{YM2018,SYM2024}.
Another interesting direction is to reinterpret the results in Section~\ref{sec:reduceentang} in terms of prior work on quantum data compression, particularly through the Koashi-Imoto decomposition~\cite{KI2002,HJPW2004}, which has a form similar to the one given in the Structure Theorem of~\cite{CHKNN2025}.
As this work opens the possibility of EA codes with bipartite states other than Bell states, it remains to be seen if new upper bounds on codes where $C\neq2^c$ can be developed, extending the work of Lai et al.~\cite{LTY2024} on semidefinite programming bounds on EA CWS codes using the techniques of Munn{\'e} et al.~\cite{MNH2024}.

Finally, throughout this work we have drawn informal analogies between the act of erasing qubits meant for the receiver with the puncturing construction from classical coding theory (see~\cite[Section 1.5]{HP2003}).
Recently, Cao et al.~\cite{CCL2025} have formalized the connection between the puncturing and shortening of stabilizer codes and the cleaning lemma by means of quantum anticodes.
While these anticodes are currently restricted to the stabilizer framework, it would be interesting if a more general version can be used to formalize the informal notion of puncturing used in this paper, and whether this leads to a better understanding of EA codes outside the stabilizer framework.

\section*{Acknowledgments}

The authors would like to thank Eric Chitambar, Markus Grassl, and Felix Huber for helpful comments.

% \bibliographystyle{IEEEtran}
% \bibliography{bibliography}


\end{document}